\documentclass[12pt,journal,draftclsnofoot,onecolumn]{IEEEtran}
\usepackage{amsfonts}
\usepackage[dvips]{graphicx}
\usepackage{epsfig,latexsym}
\usepackage{float}
\usepackage{indentfirst}
\usepackage{amsmath}
\usepackage{amssymb}
\usepackage{times}
\usepackage{subfigure}
\usepackage{stfloats,algorithm,algorithmic,bm}
\usepackage{psfrag}
\usepackage{cite}
\usepackage{subfigure}
\usepackage{color}
\usepackage{multirow}
\usepackage{chngpage}
\newtheorem{theorem}{$\mathbf{Theorem}$}
\newtheorem{lemma}{$\mathbf{Lemma}$}

\begin{document}

\title{Joint Power Splitting and Antenna Selection in Energy Harvesting Relay Channels}
\author{Zheng~Zhou, Mugen~Peng,~\IEEEmembership{Senior~Member,~IEEE}, Zhongyuan~Zhao, and Yong Li,~\IEEEmembership{Member,~IEEE}
\thanks{Copyright (c) 2012 IEEE. Personal use of this material is permitted. However, permission to use this material for any other purposes must be obtained from the IEEE by sending a request to pubs-permissions@ieee.org. Corresponding Author: Mugen Peng.}
\thanks{Zheng~Zhou (e-mail: {\tt nczhouzheng@gmail.com}), Mugen~Peng (e-mail: {\tt pmg@bupt.edu.cn}), Zhongyuan~Zhao (e-mail: {\tt zyzhao@bupt.edu.cn}), and Yong Li (e-mail: {\tt Liyong@bupt.edu.cn}) are with the Key Laboratory of Universal Wireless Communications for Ministry of Education, Beijing University of Posts and Telecommunications, China.}
\thanks{This work was supported in part by the National Natural Science Foundation of China (Grant No. 61222103), National Basic Research Program of China (973 Program) (Grant No. 2013CB336600), Beijing Natural Science Foundation (Grant No. 4131003), and Specialized Research Fund for the Doctoral Program of Higher Education (SRFDP) (Grant No. 20120005140002).}
}

\maketitle

\begin{abstract}

The simultaneous wireless transfer of information and power with the
help of a relay equipped with multiple antennas is considered in this
letter, where a ``harvest-and-forward" strategy is proposed. In
particular, the relay harvests energy and obtains information from
the source with the radio-frequent signals by jointly using the
antenna selection (AS) and power splitting (PS) techniques, and then
the processed information is amplified and forwarded to the
destination relying on the harvested energy. This letter jointly
optimizes AS and PS to maximize the achievable rate for the proposed
strategy. Considering the joint optimization is according to the
non-convex problem, a two-stage procedure is proposed to determine
the optimal ratio of received signal power split for energy
harvesting, and the optimized antenna set engaged in information
forwarding. Simulation results confirm the accuracy of the two-stage
procedure, and demonstrate that the proposed ``harvest-and-forward"
strategy outperforms the conventional amplify-and-forward (AF)
relaying and the direct transmission.
\end{abstract}

\begin{keywords}
Energy harvesting, antenna selection, power splitting.
\end{keywords}

\newpage

\vspace*{-0.5em}
\section{Introduction}
\vspace*{-0.5em}

Energy harvesting is a promising solution to increase the life cycle
of wireless nodes and hence alleviates the energy bottleneck of
green wireless networks. As an alternative to conventional energy
harvesting techniques, simultaneous wireless information and power
transfer~(SWIPT), which relies on the usage of radio frequency (RF) signals, is expected to bring some fundamental changes to the design of wireless communication networks~\cite{b10}. Considering that a
wireless relay is not always convenient to be equipped with fixed power supplies, the energy harvesting relay with SWIPT has been presented
recently, where power splitting (PS) and time switching (TS) are two
advanced protocols~\cite{b13}. PS splits the received signal power
at the relay into two parts, one is for information processing, and
the other is for energy harvesting to power the forwarding of the processed information. While the relay utilizes different time blocks to realize these two operations separately in the TS
protocol. Further on, thanks to the spatial processing in wireless nodes with multiple
antennas, the TS protocol has been extended and the information
processing and energy harvesting can be separated at different antennas over the same time, and antennas switch between decoding and rectifying
based on a antenna selection (AS) scheme~\cite{b18}.

Intuitively, in the multiple-antenna scenario, the joint PS and AS
design can reach a flexible utilization of the received RF signals,
which provides better performances than the separated PS or AS does.
Through AS, partial antennas are selected out only for energy
harvesting, and the remaining antennas are specified for both
information processing and energy harvesting, which can be optimized
further by PS. Unfortunately, few works have discussed the joint PS
and AS design for the relay with multiple antennas. In this letter,
a ``harvest-and-forward" strategy is proposed to improve the
achievable rate in energy harvesting relay channels with multiple
antenna configurations. Further, the achievable rate maximization
problem through the joint AS and PS optimization is formulated and
derived.

\vspace*{-0.5em}
\section{System Model}

An energy harvesting relay channel consisting of a single-antenna
source $S$, a single-antenna destination $D$, and a multiple-antenna
relay $R$, is depicted in Fig. \ref{fig1}(a). Both $S$ and $D$ are
devices without energy harvesting and have continuous supply of
power. $R$ is an energy harvesting device and relies on its
harvested energy to participate in signal transmission and
processing. To complete the information delivering from $S$ to $D$
with the help of $R$, a two-phase ``harvest-and-forward" strategy is
presented in this letter. Specifically, in the first phase, $R$
receives a signal from $S$, and harvests energy from a part of its
signal power. In the second phase, relying on the harvested energy,
$R$ amplifies and forwards remnant signals to $D$. For simplicity, the transmission duration of each phase is set to be normalized, thus the terms ``energy" and ``power" can be used equivalently~\cite{b12}. Besides, the power consumption for the signal receiving is assumed to be
negligible, and the harvested energy at $R$ is only used for signal forwarding~\cite{b12}.
\begin{figure}[!htb]
\centering\vspace*{-1em}
\includegraphics[width = 5 in]{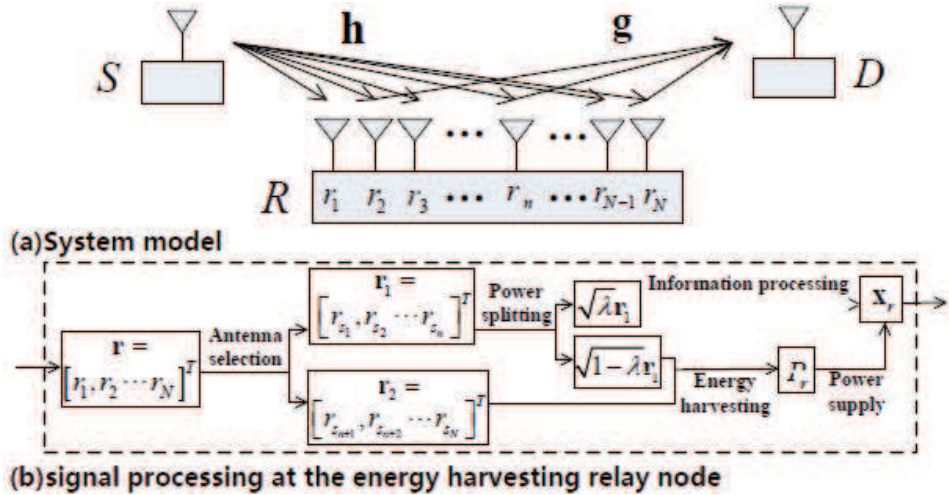}
\caption{System model and signal processing at the energy harvesting relay}\label{fig1}\vspace*{-1em}
\end{figure}

In the first phase, the received signal at $R$ with $N$ antennas can be denoted by
\begin{equation}\label{eq:1}
{\bf{r}} = \sqrt {P_s} {\bf{h}}s + {{\bf{z}}_a},
\end{equation}
where ${\bf{r}} = {\left[ {{r_{1}},{r_{2}} \cdots {r_{N}}}
\right]^T}$ is an $N \times 1$ vector, $s$ denotes the transmitted
signal from $S$ with normalized power $E({\left| s \right|^2}) = 1$,
and $P_s$ is the transmit power at $S$ with $P_s \le P$. Besides, ${\bf{h}}= {\left[ {{h_{1}},{h_{2}} \cdots {h_{N}}}
\right]^T}$ denotes
the $N \times 1$ channel vector between $S$ and $R$, and
${{\bf{z}}_a}\sim CN\left( {{\bf{0}},\sigma _a^2{{\bf{I}}_N}}
\right)$ is the $N \times 1$ additive noise vector introduced by the receiver antennas at $R$~\cite{b13}.

For the realization of the harvest-and-forward strategy, $N$ antennas are divided
into two sets via the AS technique, thus the received signal
$\bf{r}$ is split into two sub-signals (i.e., ${\bf{r_1}}$ and
${\bf{r_2}}$). The components in the first antenna set are used to
forward signals and harvest energy via the PS technique. The
ratio of sub-signal power split for the information processing is
denoted as $\lambda\in [0,1]$, and the energy harvesting is denoted
as $(1-\lambda)$. The components in the second antenna set are used to harvest energy solely. Consequently, the harvested energy from
$\bf{r}$ can be calculated as
\begin{equation}\label{eq:2}
{P_r} =\eta \left( {1 - \lambda } \right)\left\| {\bf{r_1}} \right\|_2^2 + \eta \left\| {\bf{r_2}} \right\|_2^2= \eta \left( {1 - \lambda } \right)\sum\limits_{i = 1}^n {{{\left|
{{r_{{s _i}}}} \right|}^2}}+ \eta \sum\limits_{i = n + 1}^N
{{{\left| {{r_{{s _i}}}} \right|}^2}},
\end{equation}
where $\eta \in (0,1]$ denotes the energy conversion efficiency from
signal power to circuit power, and ${\bf{r_1}} ={\left[
{{r_{s_{1}}},{r_{s_{2}}} \cdots {r_{s_{n}}}} \right]^T}$ presents
the sub-signal received by the first antenna set, where $n \in
[1,N]$ is the number of antennas therein, and ${s _1},{s _2} \cdots
{s _n}$ are labels of them.
${\bf{r_2}}={\left[{{r_{s_{n+1}}},{r_{s_{n+2}}} \cdots {r_{s_{N}}}}
\right]^T}$ describes the sub-signal received by the second
antenna set, and ${s _{n+1}},{s _{n+2}} \cdots {s _N}$ are labels of
antennas therein. Note that these variables and vector sets~(i.e, $n$, $\lambda$, and ${\bf{r_1}}$, ${\bf{r_2}}$) should be optimized for performance improvement. What's more, $P_r$ is limited, since $P_r \le \eta \left\| {\bf{r}} \right\|_2^2 \le \eta P_s$.

During the second phase, the remnant signals will be amplified and forwarded to $D$, which is powered by the harvested energy in the first phase. Note that a distributed beamforming design is adopted at $R$ as in~\cite{b20}, since the joint optimization of a centralized beamforming, PS and AS is too complex to achieve a tractable solution. In this way, we only focus on the joint optimization of PS and AS in this letter. The processed signal at $R$ is formulated as
\begin{equation}\label{eq:3}
{\bf{x_r}} ={\bf{e}}^{{\bf{j\theta }}} \beta (\sqrt{\lambda}
\bf{r_1}+ {\bf{z}}_{b} ),
\end{equation}
where ${\bf{x_r}}$, ${\bf{z}}_b$ and $\sqrt{\lambda} \bf{r_1}$ are $n \times 1$ vectors, ${\bf{z}}_b\sim CN\left( {{\bf{0}},\sigma _b^2{{\bf{I}}_n}}
\right)$ is the additive noise vector
introduced by signal conversion from passband to baseband~\cite{b13}. The harvested energy allocation is based on the strength of remnant signals at each antenna. Thus, the relay amplification gain is depicted by
\begin{equation}\label{eq:4}
\beta = \sqrt {\frac{{{P_r}}}{{\lambda \sum\limits_{i = 1}^n {\left|
{{r_{{s _i}}}} \right|}^2 + n \sigma_b^2}}}.
\end{equation}

Besides, ${{\bf{e}}^{{\bf{j\theta }}}}$ is the $n \times n$
distributed beamforming diagonal matrix~\cite{b20}, which has
${e^{j{\theta _i}}}, (i=1,2 \cdots n)$'s on the main diagonal and
zeros elsewhere, and ${\theta _{{i}}} = - \left( {\arg {h_{{s_i}}} +
\arg {g_{{s_i}}}} \right)$. The above processes are illustrated in Fig.~\ref{fig1}(b). The received signal at $D$ is expressed
as
\begin{equation}\label{eq:5}
y = {\bf{g}} {\bf{x_r}} + z = \beta \sqrt \lambda  \sqrt {P_s} \sum\limits_{i = 1}^n {\left| {{g_{{s _i}}}{h_{{s _i}}}} \right|} s + \beta {{\bf{g}}}{{\bf{e}}^{{\bf{j\theta }}}}\left( {\sqrt \lambda  {{\bf{z}}_a}' + {{\bf{z}}_b}} \right) +
 z,
\end{equation}
where ${\bf{g}}^T=[g_{s_1}, g_{s_2} \dots g_{s_n}]^T$ is the $n \times
1$ channel vector between these forwarding antennas and $D$,
${{\bf{z}}_a}'$ is the noise vector from these $n$ antennas comprising of $s_1$-th, $s_2$-th $\dots$ $s_n$-th items from ${{\bf{z}}_{a}}$, and $z\sim CN\left( {0,{\sigma ^2}} \right)$ is the Gaussian noise at $D$.

Accordingly, the received signal-to-noise-ratio (SNR) at $D$ can be
given by
\begin{equation}\label{eq:6}
\begin{split}
&SNR = \frac{{{\beta ^2}\lambda P_s{{\left( {\sum\limits_{i = 1}^n
{\left| {{g_{{s _i}}}{h_{{s _i}}}} \right|} } \right)}^2}}}{{{\beta
^2}\sum\limits_{i = 1}^n {{{\left| {{g_{{s _i}}}} \right|}^2}}
\left( {\lambda \sigma _a^2 + \sigma _b^2} \right) + {\sigma ^2}}}.
\end{split}
\end{equation}

Consequently, an optimization problem is formulated for the proposed
harvest-and-forward strategy to maximize the achievable rate:
\begin{equation}\label{eq:7}
\begin{array}{l}
(P1): \mathop {\max }\limits_{\lambda ,n,{s _1} \cdots {s _n}} \quad R = W\log_2(1+SNR),\\
s.t.\quad  \quad 0 \le \lambda  \le 1,\\
\quad \quad \quad {P_s} \le P,\\
\quad \quad \quad n,i,j,s_i,s_j \in \left[ {1:N} \right],
\end{array}
\end{equation}
where $W$ is the channel bandwidth, $n,{s _1} \cdots {s _n}$ are
variables determined by the AS scheme, and $\lambda$ is the variable
determined by the PS scheme. Note that ${P_s} = P$ must be satisfied for the achievable rate maximization.

\vspace*{-0.5em}
\section{Achievable Rate Optimization}

Note that Problem (P1) is equivalent to the SNR maximization problem~\cite{b20}. Since the above two kinds of variables are coupled together for calculating SNR in \eqref{eq:6}, the SNR maximization problem is non-linear and non-convex. To make this problem tractable, a general two-stage optimization procedure is proposed, where the antenna set is fixed firstly to determine the optimal $\lambda$, followed by the optimal antenna selection configuration.

\vspace*{-0.5em}
\subsection{Power Splitting and Optimal $\lambda$}

According to \eqref{eq:4}, \eqref{eq:6} and \eqref{eq:7}, the corresponding SNR
maximization problem with a given antenna set is reformulated as
\begin{equation}\label{eq:8}
\begin{split}
(P2): &\mathop {\max }\limits_\lambda \quad J_{\lambda}=\frac{{\eta P{B_{{\Omega _n}}}\lambda \left( {{R_N} - {R_{{\Omega _n}}}\lambda } \right)}}{{\eta {A_{{\Omega _n}}}\left( {\lambda \sigma _a^2 + \sigma _b^2} \right)\left( {{R_N} - {R_{{\Omega _n}}}\lambda } \right) + {\sigma ^2}\left( {n\sigma _b^2 + {R_{{\Omega _n}}}\lambda } \right)}},\\
&s.t.\quad \quad 0 \le \lambda \le 1,
\end{split}
\end{equation}
where ${\Omega _n}=[s_1, s_2 \dots s_n]$ denotes the set of $n$ forwarding antennas, ${A_{\Omega _n}} = \sum\limits_{i
= 1}^n {{{\left| {{g_{{s_i}}}} \right|}^2}}$, ${B_{\Omega _n}} =
{\left( {\sum\limits_{i = 1}^n {\left|{{g_{{s_i}}}{h_{{s_i}}}}
\right|} } \right)^2}$, ${R_{\Omega _n}} = \sum\limits_{i = 1}^n
{{{\left| {{r_{{s_i}}}} \right|}^2}}$, and ${R_N} = \sum\limits_{i =
1}^N {{{\left| {{r_{{s_i}}}} \right|}^2}}$.

\begin{lemma}
With a given antenna set ${\Omega _n}$, the received SNR $J_{\lambda}$ is a concave function in terms of the power splitting ratio $\lambda$.
\end{lemma}

\begin{proof}
Please refer to Appendix A.
\end{proof}

Based on Lemma $1$, Problem $(P2)$ can be regarded as a convex
optimization problem. The Karush-Kuhn-Tucker (KKT) conditions are
employed for achieving the optimal solution, which can be readily
derived by using the following theorem.

\begin{theorem}
The optimal power splitting ratio $\lambda$ under condition of a given antenna set $\Omega_n$ to maximize SNR for the proposed harvest-and-forward strategy in the multiple-antenna relay channel can be deduced by
\begin{equation}\label{eq:12}
\begin{split}
&\lambda _{{\Omega _n}}^{opt} = \min \left\{ {\lambda _{{\Omega _n}}^ * ,1} \right\},\\
&\lambda _{{\Omega _n}}^ *  = \left\{ \begin{array}{l}
\; \quad \quad \quad \quad \quad \quad \quad \quad \quad \quad \frac{{{R_N}}}{{2{R_{{\Omega _n}}}}}\quad \quad \quad \quad \quad \quad \quad \quad \quad \quad \quad \quad ;\eta \sigma _b^2{A_{{\Omega _n}}} = {\sigma ^2}\\
\sqrt {\eta \sigma _b^2{A_{{\Omega _n}}}{R_N} + n{\sigma ^2}\sigma _b^2} \frac{{\sqrt {\eta \sigma _b^2{A_{{\Omega _n}}}{R_N} + n{\sigma ^2}\sigma _b^2}  - \sqrt {{\sigma ^2}{R_N} + n{\sigma ^2}\sigma _b^2} }}{{{R_{{\Omega _n}}}\left( {\eta \sigma _b^2{A_{{\Omega _n}}} - {\sigma ^2}} \right)}}\quad \quad ;\eta \sigma _b^2{A_{{\Omega _n}}} \ne {\sigma ^2}
\end{array} \right .
\end{split}
\end{equation}
\end{theorem}

\begin{proof}
Please refer to Appendix B.
\end{proof}

\vspace*{-0.5em}
\subsection{Exhaustive Searching Based Antenna Selection}

Similar to the derivation of (8), for each feasible antenna set $\Omega _n$, substituting (4) and (9) into (6), the SNR expression can be reformulated as
\begin{equation}\label{eq:21}
J_{\Omega _n} = \frac{{\eta P{B_{{\Omega _n}}}\left( {{R_N} - \lambda _{{\Omega _n}}^{opt}{R_{{\Omega _n}}}} \right)\lambda _{{\Omega _n}}^{opt}}}{{ {\eta {A_{{\Omega _n}}}} {\left( {{R_N} - \lambda _{{\Omega _n}}^{opt}{R_{{\Omega _n}}}} \right)\left( {\lambda _{{\Omega _n}}^{opt}\sigma _a^2 + \sigma _b^2} \right) + {\sigma ^2}\left( {n\sigma _b^2 + \lambda _{{\Omega _n}}^{opt}{R_{{\Omega _n}}}} \right)}}}.
\end{equation}

Considering the fact that the number of $\Omega_n$'s is finite, the maximization of \eqref{eq:21} can be solved through searching the
optimal one among all feasible antenna sets with the determined PS
ratio.

\vspace*{-0.5em}
\subsection{Greedy Antenna Selection}

The exhaustive searching process for the optimal antenna
set $\Omega^{opt}$ can be categorized as an Non-deterministic
Polynomial-time (NP)-hard problem because the number of feasible
antenna sets is ${2^N}-1$. A greedy antenna selection scheme is presented to approach the optimal solution, which is of a complexity of $O(N^2)$, and thus easier to handle. We use $\Phi_n=\{i_n\in[1,N]| i_n\notin \Omega^{opt}\}$ to denote an antenna set with energy harvesting solely, and $i_n$'s are antennas therein. We use $\Omega_n$ to denote a feasible antenna set with power splitting, which also serves in signal forwarding. Note that $\vert\Omega_n\vert=n$ and thus $\vert\Phi_n\vert=N-n$, where
$\vert \cdots \vert$ denotes the cardinality of a set. The key idea is to determine whether there are received SNR gains when an antenna $i_n$ is switched from energy harvesting set~(i.e., $\Phi_n$) to signal forwarding set~(i.e., $\Omega_n$). Thus, ${\Omega_n}= {\Omega^{opt}} \cup [i_{n}]$, and $\Phi_n=\Phi_{n-1} / [i_n]$, where $/ $ denotes the subtraction of sets. As described in Algorithm $1$, there is no forwarding antennas
initially~(i.e., $\Omega^{opt}=\emptyset$). The optimal power
splitting ratio~(i.e., $\lambda _{{\Omega_n}}^{opt}$) and the
achieved SNR~(i.e., $J_{\Omega _n}$) for each
feasible antenna set $\Omega _n$ are calculated. Then the largest SNR with
$n$ forwarding antennas is determined by
\begin{equation}\label{eq:13}
{J'}\left( n \right) = \mathop {\max }\limits_{{\Omega _n}} \quad J_{\Omega _n}.
\end{equation}

The optimal $\Omega _n$ derived from \eqref{eq:13} is settled for signal forwarding, which updates $\Omega^{opt}$. This algorithm ends until the above procedure cannot increase the SNR performance or all antennas are settled for signal forwarding. Note that no iteration is included in our algorithm, since the greedy AS procedure is given one shot. What's more, the algorithm can converge, since it's a one-time non-decreasing SNR based searching, and the number of feasible solutions is finite. TABLE $1$ shows the complexity comparison between the proposed strategy and the exhaustive searching method.
\begin{table}[!hbp]
\centering
\caption{Complexity Comparison}
\begin{tabular}{c|c|c}
\hline
 & Proposed strategy & Exhaustive searching method \\
\hline
Time complexity & $O(N^2)$ & $2^N-1$ \\
\hline
\end{tabular}
\end{table}

\begin{algorithm}
\caption{Joint optimization of the AS and PS} \label{alg:Framwork}
\begin{algorithmic}[1]
\STATE Set the stage as $n=1$, and the optimal antenna set for signal forwarding as $\Omega^{opt} = \emptyset$.
\STATE For all~$i_n \in \Phi_n $\\
~~Set a feasible antenna set ${\Omega_n}= {\Omega^{opt}} \cup [i_{n}]$;\\
~~Calculate $\lambda _{{\Omega_n}}^{opt}$ according to \eqref{eq:12}, and $J_{\Omega _n}$ according to \eqref{eq:21}.
\STATE Derive $J'\left( n \right)$ and the optimal ${\Omega_n}$ according to \eqref{eq:13}.
\STATE If $J' \left( n \right) \ge J' \left( {n-1} \right)$, mark the optimal ${\Omega_n}$ as ${\Omega^{opt}}$, and set $n=n+1$.
\STATE Stop if $J' \left( n \right) < J' \left( n-1\right)$, or $n=N+1$, otherwise go to step $2$.
\end{algorithmic}
\end{algorithm}

\vspace*{-0.5em}
\section{Numerical Results}

The path loss model for the energy harvesting relay channel is
denoted by ${{\vert\rho_i\vert}^2}d_i^{-2}$, where $i=1,2,3$,
$d_1(d_2)$ is the distance between $S(D)$ and $R$, and $d_3$ is the
distance between $S$ and $D$. Besides, ${\vert\rho_i\vert}$ denotes
the short-term channel fading, and is assumed to be Rayleigh
distributed. ${{\vert\rho_i\vert}^2}$ follows the exponential
distribution with unit mean. We set the energy conversion efficiency
as $\eta=0.2$, and the noises as $\sigma^2=-50$ dBm,
$\sigma_a^2=\sigma_b^2=\sigma^2/2$. In addition, the bandwidth is
$W=1$ MHz, and the number of antennas is $N=10$.

Achievable rate performances for different transmission strategies are evaluated in Fig. 2. The conventional amplify-and-forward (AF) relaying and the direct transmission are traditional information transmissions without energy harvesting processes. The consumed power of the system is assumed as $P = 10 dBW$. Thus, the transmit power is $P$ at the source node for the proposed strategy and the direct transmission, and is $P/2$ at both the source node and the relay node for the conventional AF relaying. The distance between the source and the relay is normalized as $d_1 = 1$ meter. The results imply that the proposal enjoys a better achievable rate than the direct
transmission and the conventional AF relaying.

\begin{figure}[!htb]
\centering\vspace*{-1em}
\includegraphics[width= 4.5in]{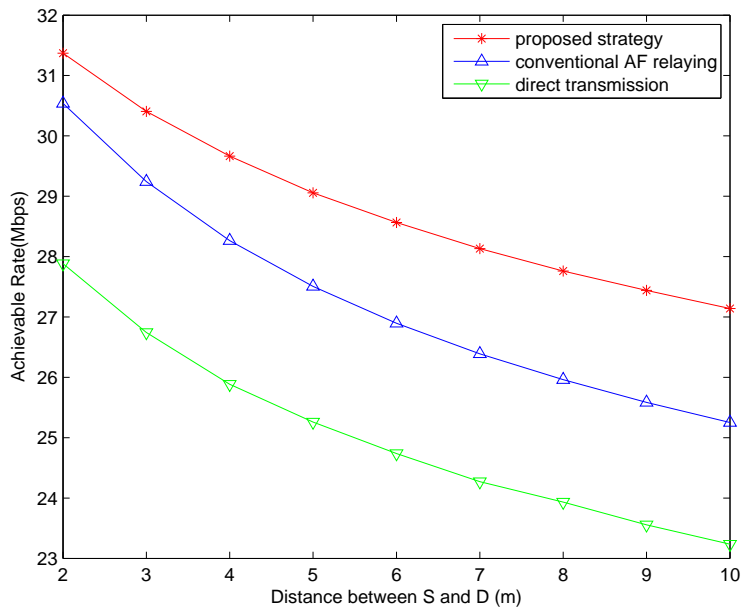}
\caption{Achievable rate versus distance between $S$ and $D$, $d_1=1,N=10,P=10dBW,\sigma^2=-50dBm$.}\label{fig3}\vspace*{-1em}
\end{figure}

Fig. \ref{fig4} is given to demonstrate the efficiency of the proposed joint optimization, where $d_1=5$ meters, $d_2=10$ meters, $d_3=15$
meters. The exhaustive searching method is to find the optimal joint AS and PS solution numerically with the help of Theorem $1$. The energy harvesting strategy with pure PS is a special case of the proposed strategy, where all antennas are selected for signal forwarding~(i.e., $n=N$). The pure AS strategy is another special case of the proposed strategy, where the sub-signal power at the selected transmitting antennas is used solely for information processing~(i.e., $\lambda=1$). It's obvious that the performance of the proposed strategy approaches to be optimal, which indicates that the
proposal is accurate and efficient. What's more, the proposed strategy outperforms the pure AS strategy or the pure PS strategy in the achievable rate. It reveals that both AS and PS techniques are indispensable to optimize the achievable rate performance of the proposed harvest-and-forward strategy.

\begin{figure}[!htb]
\centering\vspace*{-1em}
\includegraphics[width= 4.5in]{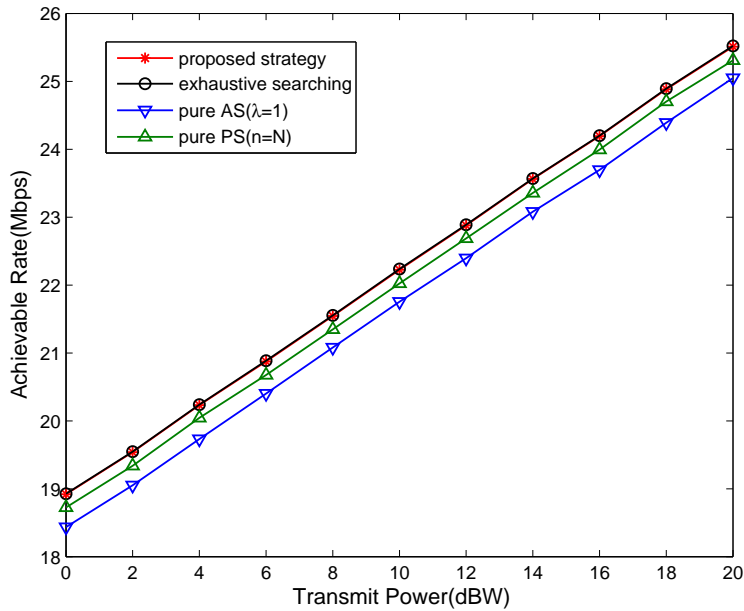}
\caption{Achievable rate versus transmit power,
$d_1=5,d_2=10,d_3=15$.}\label{fig4}\vspace*{-1em}
\end{figure}

\vspace*{-0.25em}
\section{CONCLUSIONS}
\vspace*{-0.25em}

A harvest-and-forward strategy in relay channels with multiple
antenna configurations has been proposed in this paper, where the
optimization problem in terms of the achievable rate has been solved
through jointly designing antenna selection and power splitting
techniques. Simulation results have indicated that the proposed
strategy and the corresponding solution have significant achievable
rate performance gains. The optimization of energy efficiency
performance when jointly considering power splitting and antenna
selection would be analyzed in the future.

\begin{appendices}
\vspace*{-0.5em}
\section{PROOF OF LEMMA 1}
\vspace*{-0.25em}

From \eqref{eq:8}, the second-order derivative of $J_{\lambda}$ is derived as
\begin{equation}\label{eq:15}
\begin{array}{l}
\frac{{{\partial ^2}J_{\lambda}}}{{\partial {\lambda ^2}}} = \frac{{ - 2{C_{{\Omega _n}}}}}{{\eta {A_{{\Omega _n}}}\left( {{R_N} - {R_{{\Omega _n}}}\lambda } \right)\left( {\sigma _a^2\lambda  + \sigma _b^2} \right) + {\sigma ^2}\left( {{R_{{\Omega _n}}}\lambda  + n\sigma _b^2} \right)}}\\
*\left\{ {\left( {{\sigma ^2}{R_{{\Omega _n}}}{R_N} + n{\sigma ^2}\sigma _b^2{R_{{\Omega _n}}}} \right)\left( {\eta \sigma _b^2{A_{{\Omega _n}}}{R_N} + n{\sigma ^2}\sigma _b^2} \right)} \right. + \eta \sigma _a^2{A_{{\Omega _n}}}\left[ {\eta \sigma _b^2{A_{{\Omega _n}}}} \right.\\
*{\left( {{R_N} - {R_{{\Omega _n}}}\lambda } \right)^3}\left. { + {\sigma ^2}{{\left( {{R_{{\Omega _n}}}\lambda } \right)}^3}} \right]\left. { + \eta \sigma _a^2{A_{{\Omega _n}}}n{\sigma ^2}\sigma _b^2\left[ {{{\left( {{R_N} - \frac{3}{2}{R_{{\Omega _n}}}\lambda } \right)}^2} + \frac{3}{4}{{\left( {{R_{{\Omega _n}}}\lambda } \right)}^2}} \right]} \right\},
\end{array}
\end{equation}
where ${C_{{\Omega _n}}} = {{\eta P{B_{{\Omega _n}}}}
\mathord{\left/
 {\vphantom {{\eta P{B_{{\Omega _n}}}} {\left[ {\left( {\eta \sigma _a^2{A_{{\Omega _n}}}{R_{{\Omega _n}}}} \right){\lambda ^2} - \left( {\eta {A_{{\Omega _n}}}\sigma _a^2{R_N}} \right.} \right.}}} \right.
 \kern-\nulldelimiterspace} {\left[ {\left( {\eta \sigma _a^2{A_{{\Omega _n}}}{R_{{\Omega _n}}}} \right){\lambda ^2} - \left( {\eta {A_{{\Omega _n}}}\sigma _a^2{R_N}} \right.}  \left. { - \eta {A_{{\Omega _n}}}\sigma _b^2{R_{{\Omega _n}}} + {\sigma ^2}{R_{{\Omega _n}}}} \right)\lambda \right. - }}\\
{\left.{ \sigma _b^2\left( {\eta {A_{{\Omega _n}}}{R_N} + n{\sigma ^2}} \right)} \right]^2} \ge 0$.

Since $\lambda \le 1$, and ${{R_{{\Omega _n}}} \le {R_N}}$, we have
${{R_N} - {R_{{\Omega _n}}}\lambda} \ge 0$. Therefore,
$\frac{{{\partial ^2}J_{\lambda}}}{{\partial {\lambda ^2}}} \le 0$, and $J_{\lambda}$ is
a concave function of $\lambda$. This completes the proof of Lemma
1.

\vspace*{-0.5em}
\section{PROOF OF THEOREM 1}
\vspace*{-0.5em}

Since (P2) is concave, and the feasible set for $\lambda$ is convex,
the KKT conditions are sufficient for achieving the optimal solution with the Lagrange function
\begin{equation}\label{eq:9}
L(\lambda, \mu) = J_{\lambda} - {\mu}\left( {\lambda  - 1} \right),
\end{equation}
where ${\mu} \ge 0$ is the Lagrange multiplier associated with the
constraint $\lambda-1 \le 0$. The KKT conditions are stated by
\begin{equation}\label{eq:10}
\begin{array}{l}
\frac{{\partial L(\lambda, \mu)}}{{\partial \lambda }} = \frac{{\partial J_{\lambda}}}{{\partial \lambda }} - {\mu} = 0,\\
{\mu}\left( {\lambda  - 1} \right) = 0,\\
\lambda  - 1 \le 0,
\end{array}
\end{equation}
where $\frac{{\partial J_{\lambda}}}{{\partial \lambda }}$ is the first-order
derivative of $J_{\lambda}$, and is given by
\begin{equation}\label{eq:11}
\begin{array}{l}
\frac{{\partial J_{\lambda}}}{{\partial \lambda }} = {C_{{\Omega _n}}}\left[ {{{\left( {{R_{{\Omega _n}}}} \right)}^2}\left( {\eta \sigma _b^2{A_{{\Omega _n}}} - {\sigma ^2}} \right){\lambda ^2} - 2{R_{{\Omega _n}}}\left( {\eta \sigma _b^2{A_{{\Omega _n}}}{R_N} + n{\sigma ^2}\sigma _b^2} \right)\lambda } \right.\\
\quad \quad \quad \quad \quad \quad \quad \quad \quad \quad \quad \quad \quad \quad \quad \quad \quad \quad \quad \quad \left. { + {R_N}\left( {\eta \sigma _b^2{A_{{\Omega _n}}}{R_N} + n{\sigma ^2}\sigma _b^2} \right)} \right].
\end{array}
\end{equation}

There are two groups of solutions for the KKT conditions \eqref{eq:10}. First, $\lambda_1 = 1$, and ${\mu} = {\left. {\frac{{\partial J_{\lambda }}}{{\partial \lambda }}} \right|_{\lambda  = 1}}$. Second, $0 \le \lambda < 1$, and ${\mu} =0$. When ${\mu} =0$ and $(\eta \sigma _b^2{A_{{\Omega _n}}} - {\sigma^2}) \ne 0$, it's derived that
\begin{equation}\label{eq:16}
\begin{array}{l}
\lambda_2=\sqrt {\eta \sigma _b^2{A_{{\Omega _n}}}{R_N} + n{\sigma ^2}\sigma _b^2}\frac{{\sqrt {\eta \sigma _b^2{A_{{\Omega _n}}}{R_N} + n{\sigma ^2}\sigma _b^2}  + \sqrt {{\sigma ^2}{R_N} + n{\sigma ^2}\sigma _b^2} }}{{{R_{{\Omega _n}}}\left( {\eta \sigma _b^2{A_{{\Omega _n}}} - {\sigma ^2}} \right)}},\\
\lambda_3=\sqrt {\eta \sigma _b^2{A_{{\Omega _n}}}{R_N} + n{\sigma ^2}\sigma _b^2}\frac{{\sqrt {\eta \sigma _b^2{A_{{\Omega _n}}}{R_N} + n{\sigma ^2}\sigma _b^2}  - \sqrt {{\sigma ^2}{R_N} + n{\sigma ^2}\sigma _b^2} }}{{{R_{{\Omega _n}}}\left( {\eta \sigma _b^2{A_{{\Omega _n}}} - {\sigma ^2}} \right)}},
\end{array}
\end{equation}
whose values depend on $(\eta \sigma _b^2{A_{{\Omega_n}}} - {\sigma ^2})$. It's clear that $\lambda _2 < 0$ or ${\lambda _2} > 1$, thus $\lambda_2$ is not feasible. When ${\mu} =0$ and
$(\eta \sigma _b^2{A_{{\Omega _n}}} - {\sigma ^2}) = 0$, it's derived from \eqref{eq:10} that $\lambda_4= R_N/{2R_{\Omega_n}}$. Next, the optimal power splitting ratio $\lambda^{opt}_{\Omega_n}$ is determined through monotonicity analysis of the optimization object $J_{\lambda}$. According to \eqref{eq:11}, when $(\eta \sigma _b^2{A_{{\Omega _n}}} - {\sigma ^2}) \ne0$, the expression can be factorize as
\begin{equation}\label{eq:17}
\frac{{\partial J_{\lambda}}}{{\partial \lambda }} = {C_{{\Omega _n}}}  {\left( {{R_{{\Omega _n}}}} \right)^2} \left( {\eta \sigma _b^2{A_{{\Omega _n}}} - {\sigma ^2}} \right)  (\lambda - \lambda_2)  (\lambda - \lambda_3).
\end{equation}

Since $0 \le {C_{{\Omega_n}}} {\left( {{R_{{\Omega _n}}}}
\right)^2}$ and $\left( {\eta\sigma _b^2{A_{{\Omega _n}}} - {\sigma
^2}} \right)  (\lambda -\lambda_2) \le 0$, when $\lambda_3 < 1$,
$\frac{{\partial J_{\lambda}}}{{\partial \lambda }}$ is non-negative among
$[0,\lambda_3]$, and is non-positive among $(\lambda_3,1]$, i.e.,
$\lambda^{opt}_{\Omega_n}= \lambda_3$. In the case of $\lambda_3 \ge
1$, $\frac{{\partial J_{\lambda}}}{{\partial \lambda }}$ is non-negative among
$[0,1]$, i.e., $\lambda^{opt}_{\Omega_n} = 1$. In conclusion,
$\lambda^{opt}_{\Omega_n} = min \{1, \lambda_3 \}$. In a similar way, it can be derived that $\lambda^{opt}_{\Omega_n} = min \{1, \lambda_4
\}$, when $(\eta
\sigma _b^2{A_{{\Omega _n}}} - {\sigma ^2}) = 0$. Consequently, the optimal $\lambda$ is derived as in
\eqref{eq:12}. This proves Theorem 1.
\end{appendices}

\end{document}